\documentclass[runningheads,a4paper]{llncs}
\usepackage{amssymb,amsmath}
\usepackage{graphicx}
\begin{document}

\title{Multi-layered graph-based multi-document summarization model}
\title{Multi-layered graph-based multi-document summarization model}
\author{Ercan Canhasi}         
\institute{ University of Prizren\\ Faculty of Computer Science\\  Rruga e Shkronjave” nr.1 20000 Prizren, Republika e Kosovës}
\date{May 21, 2014}    
\toctitle{Lecture Notes in Computer Science}
\tocauthor{Authors' Instructions}
\maketitle

\begin{abstract}

Multi-document summarization is a process of automatic generation of a compressed version of the given collection of documents. Recently, the graph-based models and ranking algorithms have been actively investigated by the extractive document summarization community. While most work to date focuses on homogeneous connecteness of sentences and heterogeneous connecteness of documents and sentences (e.g. sentence similarity weighted by document importance), in this paper we present a novel 3-layered graph model that emphasizes not only sentence and document level relations but also the influence of under sentence level relations (e.g. a part of sentence similarity).

\end{abstract}

\section{Introduction}
Multi-document summarization (MDS) aims to filter the most important information from a set of documents to generate a compressed summary. Recently, the graph-based models and ranking algorithms have been extensively researched. \\ \indent While most work to date focuses on the sentence and the document level relations, in this work we present a novel 3-layered graph model that emphasizes not only the sentence and the document level relations but also the influence of the under sentence level relations. The document set D={$d_1$,$d_2$,\ldots$d_n$} is represented as a weighted undirected frame graph G. As a difference to previous works in our graphs there is not only one kind of objects (i.e. sentences), but there are three kinds of objects: semantic role frames, sentences and documents. \\ \indent Even if humans do not always agree on the content to be added to a summary, they perform very well on this task. Therefore our goal should be to find a way of mimicking the cognition behind the human like summarization process. For this challenge we consider using the psychology cognitive situation model, namely the Event-Indexing model \cite{Zwaan1}. According to this model a human-like system should keep track of five indices while reading the document. Those indices are \emph{protagonist}, \emph{temporality}, \emph{spatiality}, \emph{causality} and \emph{intention}, with the given descending order of importance. One can also show that the semantic role parser's \cite{Carreras} output can be mapped to the above proposed cognitive model. Semantic roles are defined as the relationships between syntactic constituents and the predicates. Most sentence components have semantic connections with the predicate, carrying answers to the questions such as who, what, when, where etc. \\
\indent The summarization method, we propose, works in the following way. First, the documents are given to the SRL parser where the semantic arguments from each parsed sentence are extracted. Based on the event-indexing model we calculate the composite similarity between all semantic frames. Then we generate a semantic graph where nodes are semantic frames and edges are the composite similarity values. By using an intelligent weighting scheme we add two more layers, namely the sentence and the document layers, which yields the richer multilayered graph model with the inter and intra sentence and the documental level relations. Next we use modified version of PageRank for identifying the significant edges in the graph. The next step aims to further remove redundant information in the summary by penalizing the sentences largely overlapping with other high ranked sentences. Based on the text graph and the obtained rank scores, a greedy algorithm is applied to inflict the diversity penalty and compute the final rank scores of the sentences. Later, we sum the PageRank scores of semantic frames, originating from the same sentence, and we use it as a score for sentence scoring. Subsequently, the top scoring sentences are selected one-by-one and put into the summary.\\ \indent The remainder of this paper is organized as follows. Section 2 reviews existing graph-based summarization models. Section 3 and 4 introduces the proposed sentence ranking algorithm. After that, Section 5 reports experiments and evaluation results. Finally, Section 6 concludes the paper.

\section{Related work}
The graph-based models have been developed by the extractive document summarization community in the past years \cite{erk01,rada1}. Conventionally, they model a document or a set of documents as a text graph composed by taking a text unit as a node and similarity between text units as edges. The significance of a node in a graph is estimated by graph-based ranking algorithms, such as PageRank \cite{page} or HITS \cite{hits}. Sentences in document(s) are ranked based on the computed node significance and the most salient ones are selected to form an extractive summary. An algorithm called LexRank \cite{erk01}, adapted from PageRank, was applied to calculate sentence significance, which was then used as the criterion to rank and select summary sentences. Meanwhile, Mihalcea and Tarau \cite{rada1} presented their PageRank variation, called TextRank, in the same year.

\section{Multilayered graph model}

In this section we present our novel graph model, which will be used in frame ranking algorithm presented in the next section. Let a set of documents D be a text similarity graph $G=(V_f, V_s, V_d, E^{V_f}, E^{V_s}, E^{V_d}, \alpha_v, \beta_v, \gamma_v, \alpha_e, \beta_e, \gamma_e)$, where $V_f, V_s$  and $V_d$ represent the frame, sentence and document vertex set, respectively. $E^{V_f}\subseteq V_f \times V_f , E^{V_s}\subseteq V_s \times V_s$ and $E^{V_d}\subseteq V_d \times V_d$ are frame, sentence and document edge set. $\alpha_v:V_f\rightarrow\Re_+ , \beta_v:V_s\rightarrow\Re_+$ and $\gamma_v:V_d\rightarrow\Re_+$ are three functions defined to label frame, sentence and document vertices, while $\alpha_e:E^{V_f} \rightarrow \Re_+, \beta_e : E^{V_s} \rightarrow \Re_+$ and $\gamma_e: E^{V_d} \rightarrow \Re_+$ are functions for labeling frame, sentence and document edges.
\\
\indent Adding new layers to the classical sentence similarity graph yields the richer multilayered graph model with the inter and intra sentence and the documental level relations that can be then used to enhance the existing PageRank algorithm. Figure~\ref{MDS System}. shows the conventional text graph model before and after applying the concepts of multilayered graph representation. One can easily show that the new graph model brings to light the following previously ignored information: sentence to sentence similarity can now be distinguished in two groups, one within a document and one across two documents; the document's significance can influence the sentence ratings; there is also a completely new kind of objects (i.e. semantic role labeler (SRL) frames) involved in representing the inner sentence relations.
\begin{figure}
\centering
\includegraphics[width=0.7\textwidth]{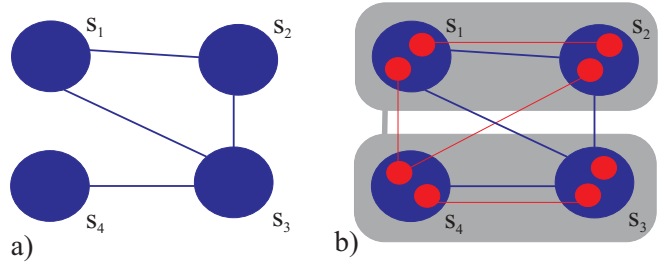}
\caption{Summarization graph model (a) before and (b) after introducing multilayered model. }
\label{MDS System}
\end{figure}
\\
\indent
The sentence edge function $\beta_e(s_i,s_j)=sim_{normal}(s_i,s_j)=\frac{sim(s_i,s_j)}{\sum_{s_k \in S \wedge k \neq i}sim(s_i,s_k)}$ and document edge function $\gamma_e(d_i,d_j)=sim_{normal}(d_i,d_j)=\frac{sim(d_i,d_j)}{\sum_{d_k \in D \wedge k \neq i}sim(d_i,d_k)}$ are formulated as the normalized similarity between the two sentences $s_i$ and $s_j$, and the two documents $d_i$ and $d_j$, respectively. The SRL frame edge function $\alpha_e(f_i,f_j)=sim_{composite}(f_i,f_j)$, is formulated as the composite similarity function of two frames $f_i$ and $f_j$. Let N be the total number of frames in a documents set. The frame vertex function $\alpha_v(f_i)=$ assigns to frame vertices the value of $1/N$ or $1/2N$, depending on their completeness, where incomplete frames have lower weight. The sentence $\beta_v(s_i)=centr\_norm(s_i)=\frac{\sum_{u \in s_i}cw(u)}{\sum_{v \in S}cw(v)}$ and document vertex $\gamma_v(d_i)=centr\_norm(d_i)=\frac{\sum_{u \in d_i}cw(u)}{\sum_{v \in D}cw(v)}$ functions are defined by the normalized centroid-based weight of the sentence and document, respectively where $cw(u)$ denotes the centroid \cite{radev} weight of the word u.
\\
\indent Our goal is to capture the similarity and redundancy between sentences, but at a lower structural and a higher semantic level. To accomplish this, we use the event-indexing model as the base for calculations of semantic similarity between frames of semantic role parser outputs, namely frames. Let us define the similarity measure for protagonist $sim_{protagonist}(f_i,f_j)=\alpha_1 \cdot sim(A0_i,A0_j) + \alpha_2 \cdot sim(A1_i,A1_j) + \alpha_3 \cdot sim(A2_i,A2_j) + \alpha_4 \cdot sim(A0_i,A1_j) + \alpha_5 \cdot sim(A0_i,A2_j)+ \alpha_6 \cdot sim(A1_i,A2_j)$; temporality $sim_{temporality}(f_i,f_j)=sim(Am\_Tmp_i,Am\_Tmp_j)$; spatiality $sim_{spatiality}(f_i,f_j)=sim(Am\_Loc_i,Am\_Loc_j)$ and causality \\ $sim_{causality}(f_i,f_j)=sim(Predicate_i,Predicate_j)$.
In order to have the flexible weighting scheme we use coefficients $\alpha_1=\alpha_2=\alpha_3=0.25; \alpha_4=\alpha_5=0.10; \alpha_6=0.5$. The compose similarity is defined as:
\begin{align}
 sim_{composite}(f_i,f_j)=\Big(\beta_1sim_{protagonist}(f_i,f_j) + \beta_2sim_{temporality}(f_i,f_j) \nonumber \\
 + \beta_3sim_{spatiality}(f_i,f_j) + \beta_4sim_{causality}(f_i,f_j)\Big)/\#arguments \nonumber
\end{align}
where $\beta_1=0.4; \beta_2=0.3; \beta_3=0.2; \beta_4=0.1 .$ The values for coefficients are chosen based on the cognitive model which gives an emphasis in the decreasing order to indices.

\section{Multilayered graph-based ranking algorithm}
In previous section the idea of multilayered text similarity graph is presented, based on it in this section we present a modified iterative graph-based sentence ranking algorithm. \\
\indent Our algorithm is  extended from those existing PageRank-like algorithms reported in the literature that calculate the graph only in the sentence level \cite{erk01,rada1}, or sentence and document level \cite{erk02,furuWei,wan}. \\ \indent In the summary, PageRank method (in matrix notation)
as described in the original paper \cite{page} is

\begin{align}
\pi^{(k+1)T} = \alpha\pi^{(k)T}\textbf{H}+(\alpha\pi^{(k)T}a+1-\alpha)\textbf{v$^T$} \nonumber
\label{pagerank}
\end{align}
where \textbf{H} is a very sparse, raw sub stochastic hyperlink matrix,  $\alpha$ is a scaling parameter between 0 and 1, $\pi^T$ is
the stationary row vector of \textbf{H} called the PageRank vector, \textbf{v}$^T$ is a complete dense, rank-one teleportation matrix and a is a binary dangling node vector. In terms of the sentence ranking the matrix \textbf{H} is an adjacency matrix of similarities, \textbf{v$^T$} is the affinity vector and the resulting $\pi^T$ is the frame ranking vector.\\ \indent
For sake of simplicity, we just assume there are two documents (e.g.$D_1,D_2$) and 4 sentences
(e.g.$S_{1,1},S_{1,2}$ in $D_1$ and $S_{2,1},S_{2,2}$ in $D_2$) involved in ranking. Yet there are eight frames extracted from four sentences, let us show the document, sentence (just the first one)
and frames (again just the first one) similarity matrices:
\\
$H_0 = \begin{bmatrix} D_{1,1}&D_{1,2}\\ D_{2,1}&D_{2,2} \end{bmatrix}$ $D_{1,1} = \begin{bmatrix} S_{1,1}&S_{1,2}\\ S_{2,1}&S_{2,2} \end{bmatrix}$  $S_{1,1} = \begin{bmatrix} F_{1,1}&F_{1,2}\\ F_{2,1}&F_{2,2} \end{bmatrix}$ \\
$H_{D_{1,1}} = \begin{bmatrix} w_{S_{11}}S_{1,1}&w_{S_{12}}S_{1,2}\\ w_{S_{21}}S_{2,1}&w_{S_{22}}S_{2,2} \end{bmatrix}$
$H_{S_{1,1}} = \begin{bmatrix} w_{F_{11}}F_{1,1}&w_{F_{12}}F_{1,2}\\ w_{F_{21}}F_{2,1}&w_{F_{22}}F_{2,2} \end{bmatrix}$
\\
\\ \indent
The block matrix $D_{1,1}$ refers to the similarity matrix of the sentences in document $D_1$, while $D_{1,2}$ represents the fold-document ($D_1$ and $D_2$) similarity matrix, and so on and so forth. Similarly the block matrix $S_{1,1}$ in $D_{1,1}$ denotes the similarity  matrix of the frames in sentence $S_1$, while $S_{1,2}$ denotes the fold-sentence ($S_1$ and $S_2$) affinity matrix. $H_0$ corresponds to the original sentence similarity matrix used in generic graph methods. The effective way of integrating
the document and sentence dimension into the $H_0$ is to highlight the document and sentence influences on the sentence and
frame edges that connect different documents and sentences as illustrated in formulas for $H_{D_{1,1}}$ and $H_{S_{1,1}}$.
The weight matrix $W_d = \begin{bmatrix} w_{S_{11}}&w_{S_{12}}\\ w_{S_{21}}&w_{S_{22}} \end{bmatrix}$ is used to distinguish
the cross-document sentence edges and the weight matrix $W_s = \begin{bmatrix} w_{F_{11}}&w_{F_{12}}\\ w_{F_{21}}&w_{F_{22}} \end{bmatrix}$ is used to distinguish the cross-sentence frame edges. The diagonal elements in $W_d$ and $W_s$  are set to 1 to neutralize the influence of the intra-document sentence and intra-sentence frame edges. Just on the opposite, the non-diagonal elements are weighted by the connections between the two corresponding documents and sentences. We define $W_d$ as $W_d(i,j)=1+ \gamma_e(d(s_i),d(s_j))$, where $d(s_i)$ presents the document that contains the sentence $s_i$. We also define $W_s$ as $W_s(i,j)=1+ \beta_e(s(f_i),s(f_j))$, and $s(f_i)$ presents the sentence that contains the frame $f_i$.
\\ \indent Based on assumption that a frame from the document and the sentence with higher significance should be ranked higher we
reflect the influence of the document and sentence dimensions on the affinity vector $\vec{v}$. Consequently, the
centroid-based weight of the document and the sentence are taken as the weights on the affinity vector $\vec{v}$.
See the following preference vectors:\\
$ \vec{v}_{o}=[\vec{v}_{d_1} \vec{v}_{d_2}]^T$;    \qquad  $\vec{v}=\bigg([\vec{v}_{d_1} \vec{v}_{d_2}]\cdot\begin{bmatrix}
w_{d_1}& \\ & w_{d_2} \end{bmatrix}\bigg)^T$; \\
$\vec{v}_{d_1}=\bigg([\vec{v}_{s_1} \vec{v}_{s_2}]\cdot\begin{bmatrix} w_{s_1}& \\ & w_{s_2} \end{bmatrix}\bigg)^T$;
$\vec{v}_{d_2}=\bigg([\vec{v}_{s_3} \vec{v}_{s_4}]\cdot\begin{bmatrix} w_{s_3}& \\ & w_{s_4} \end{bmatrix}\bigg)^T$;
\\ \indent Here $\vec{v}_{o}$ represents the original preference vector, as used in LexRank, $\vec{v}_{s_1}$ to $\vec{v}_{s_4}$ denote
the sub-preference vector of the frames from sentences $s_1$ to $s_4$, respectively;
and $\vec{v}_{d_1},\vec{v}_{d_2}$ denote the sub-preference vector of the sentences
from documents $d_1$  and $d_2$, respectively.
The weight matrices $W_{v_d} = \begin{bmatrix} w_{d_1}& \\ & w_{d_2} \end{bmatrix}$,
$W_{v_{s1}} = \begin{bmatrix} w_{s_1}& \\ & w_{s_2} \end{bmatrix}$  and $W_{v_{s2}} = \begin{bmatrix} w_{s_1}& \\ & w_{s_2} \end{bmatrix}$ are specified to introduce the bias to sentences from different documents and the bias towards frames from different sentences, respectively.
The weighting functions are defined as $ W_{v_d}(i)=1+\gamma_v(d(s_i))$, $W_{v_s}(i)=1+ \beta_v(s(f_i))$.
To ensure the solution of proposed algorithm, we should first make $\vec{v}$ a affinity probability vector.
\begin{lemma}
\label{lema1}
$\vec{v}$ is a probability vector, if $W_{v_{s_1}}$, $W_{v_{s_2}}$ and $W_{v_d}$ are positive and the diagonal elements in them sum to 1;
\end{lemma}
\begin{proof}
Since $\vec{v}_{s_1}$, $\vec{v}_{s_2}$, $\vec{v}_{s_3}$ and $\vec{v}_{s_4}$ are probability vectors of frames from
four sentences we have $|\vec{v}_{s_1}|=|\vec{v}_{s_2}|=|\vec{v}_{s_3}|=|\vec{v}_{s_4}|=1$. Then,
\begin{align*}
 |\vec{v}| &= w_{d_1}(w_{s_1}|\vec{v}_{s_1}|+w_{s_2}|\vec{v}_{s_2}|)+ w_{d_2}(w_{s_3}|\vec{v}_{s_3}| + w_{s_4}|\vec{v}_{s_4}|) \\
 & =  w_{d_1}(w_{s_1}+w_{s_2})+ w_{d_2}(w_{s_3}+w_{s_4})=1
\end{align*}
\begin{equation*}
\mbox{if } (w_{s_1}+w_{s_2} )=1 \mbox{ and } (w_{s_3}+w_{s_4})=1 \mbox{ hence } w_{d_1}+w_{d_2}=1
\end{equation*}
\end{proof}
Afterward, we should make the matrix H column stochastic and irreducible by forcing each of four block matrices of sentences and two
matrices of documents to be column stochastic simply by normalizing them by columns. To make H irreducible, we make the sixteen block matrices in H irreducible by adding additional links between any two frames, which is also adapted in PageRank.
\begin{lemma}
\label{lema2}
H is column stochastic and irreducible.
\end{lemma}
\begin{proof}
H is column stochastic since the weight matrix W is column stochastic. Let A, B, C, and D, donate any of the 4 column block in H, then
\begin{eqnarray*}
\sum_{i}H_{ij}=w_{d_{1k}}\bigg(w_{s_{1k}}\sum_{i}A_{ij}+w_{s_{2k}}\sum_{i}B_{ij}\bigg) \qquad  \qquad  \qquad\\
+w_{d_{2k}}\bigg(w_{s_{3k}}\sum_{i}C_{ij}+w_{s_{4k}}\sum_{i}D_{ij}\bigg)(k=1,\ldots,4) \\
\sum_{i}H_{ij}= w_{d_{1k}}(w_{s_{1k}}+w_{s_{2k}})+w_{d_{2k}}(w_{s_{3k}}+w_{s_{4k}}) \qquad  \qquad
\end{eqnarray*}
\begin{equation*}
\mbox{if } (w_{s_1}+w_{s_2} )=1 \mbox{ and } (w_{s_3}+w_{s_4})=1 \mbox{ hence } w_{d_1}+w_{d_2}=1
\end{equation*}
Since the four graphs corresponding to the four diagonal block matrices in H are strongly connected (i.e. they are irreducible)
and the edges connecting the four graphs are bidirectional, the graph corresponding to H is obviously strongly connected.
Thus H must be also irreducible.
\end{proof}

Notice that we must ensure $W_{v_{s_1}}>0, W_{v_{s_2}}>0$ and make the sum of the diagonal elements equal to 1 in order to ensure $\vec{v}$ to be a probability vector. And we must make H column stochastic by setting $W_d>0, W_s>0$ and both matrices column stochastic.
Finally, we obtain that H is stochastic, irreducible and primitive, hence we can compute the unique dominant
vector (with 1 as the eigenvalue) of H by using the power iteration method applied to H which converges to $\pi$.
The previous explanation is given as example with a two-document, a four sentences (two of them in every document),
and eight frames (two of frames in every single sentence).  However, we can come to the same conclusion when the number of the documents,
sentences and frames involved extends from given values to arbitrary number. \\ \indent We can finally summarize the new ranking algorithm with following two functions.
\begin{eqnarray*}
H(i,j)=\alpha_e(f_i,f_j)W_d(i,j)W_s(i,j); \qquad \\
\vec{v}=\alpha_v(f_i)(1+ \gamma_v(d(s_i)))(1+\beta_v(s(f_i))).
\end{eqnarray*}
So far the document and sentence dimensions have been integrated into the PageRank-like algorithms for frame ranking with a solid mathematical foundation.

\section{Evaluation}

The DUC\footnote[1]{Document Understanding Conference(http://duc.nist.gov)} 2004 data set from DUC was tested to analyze the efficiency of the proposed summarization method. The Task 2 at the DUC 2004 is to generate a short summary (665 bytes) of an input set of topic-related news articles.
\begin{table}
\caption{ROUGE-1 scores of the DUC 2004 and evaluation of our model }
\begin{center}
\begin{tabular}{l r}
\hline\hline
Systems & ROUGE-1 (95\% Confidence interval) \\ [0.5ex]
\hline
Avg. of human assessors & 0.403 [0.383,0.424]  \\
\\
Best machine (SYSID = 65) & 0.382 [0.369,0.395] \\
Median machine (SYSID = 138) & 0.343 [0.328,0.358] \\
Worst machine (SYSID = 111) & 0.242 [0.230,0.253] \\ \\
\textbf{Our model} & 0.379 [0.361,0.389] \\
LexRank & 0.369 [0.354,0.382] \\ \\
2 (NIST Baseline) (Rank: 25/35) & 0.324 [0.309,0.339] \\
Random baseline:& 0.315 [0.303,0.328] \\[1ex]
\hline
\end{tabular}
 \end{center}
\label{tab:results}
\end{table}
The total number of document groups is 50, with each group containing 10 articles on average. For each group, four NIST assessors were asked to create a brief summary. Machine-generated summaries are evaluated using ROUGE \cite{rouge} automatic n-gram matching which measures performance based on the number of co-occurrences between machine-generated and ideal summaries in different word units. The 1-gram ROUGE score (a.k.a.ROUGE-1) has been found to correlate very well with human judgements at a confidence level of 95\%, based on various statistical metrics. Even though in this version of method we did not consider sentence positions or other summary quality improvement techniques such as sentence reduction, its overall performance is promising, please see table~\ref{tab:results}. The use of multilayered model in summarization can make considerable improvements even though the results presented here do not report a significant difference.

\section{Conclusion and future work}
We have presented a multilayered graph model and a ranking algorithm, for generic MDS. The main contributions of our work is introducing the concept of 3-layered graph model. The results of applying this model on extractive summarization are quite promising. There is work still left to be done, however. We are now working on further improvements of the model, and it's adaptation to other summarization tasks, such as the query and update summarization.

\subsubsection*{Acknowledgments.} We would like to thank the anonymous reviewers for their constructive comments.

\bibliographystyle{plain}

\end{document}